\documentclass[thmsa]{article}
\usepackage{amssymb}
\usepackage{amsmath,amsthm,amsfonts}  
\usepackage[applemac]{inputenc}  
\usepackage{graphicx}
\usepackage{amssymb}
\usepackage{amsmath}
\usepackage{amsthm}
\usepackage{color}
\newtheorem{defi}{Definition}[section]
\newtheorem{prop}{Proposition}[section]
\newtheorem{lem}{Lemma}[section]
\newtheorem{theo}{Theorem}[section]
\newtheorem{coro}{Corollary}[section]

\def\eqdef{\stackrel{\mbox{\tiny def}}{=}}     
\newcommand{\ket}[1]{|\kern.3ex#1\kern.3ex\rangle}
\newcommand{\bra}[1]{\langle\kern.3ex #1 \kern.3ex|}
\newcommand{\scalar}[2]{\langle\kern.3ex #1 \kern.3ex|\kern.3ex#2\kern.3ex\rangle}

\begin{document}

\author{
E.M.F.  Curado$^{\mathrm{a,b,c}}$,  J. P. Gazeau$^{\mathrm{b}}$, 
 Ligia M.C.S. Rodrigues$^{\mathrm{a}}$(\footnote{e-mail:
 evaldo@cbpf.br,
gazeau@apc.univ-paris7.fr, ligia@cbpf.br} )\\
\emph{$^{\mathrm{a}}$ Centro Brasileiro de Pesquisas Fisicas,}\\
\emph{Rua Xavier Sigaud 150, 22290-180 - Rio de Janeiro, RJ, Brazil}\\
\emph{$^{\mathrm{b}}$ Laboratoire APC,
Univ Paris  Diderot, Sorbonne Paris Cit\'e, }   
\emph{75205 Paris, France} \\
\emph{$^{\mathrm{c}}$ Instituto Nacional de Ciencia e Tecnologia - Sistemas Complexos}}

\title{Generating functions for generalized binomial distributions}
\author{
H. Bergeron$^{\mathrm{a}}$, 
E.M.F.  Curado$^{\mathrm{b,c,d}}$,  
J. P. Gazeau$^{\mathrm{d}}$, 
 Ligia M.C.S. Rodrigues$^{\mathrm{b}}$(\footnote{e-mail: herve.bergeron@u-psud.fr, 
 evaldo@cbpf.br,
gazeau@apc.univ-paris7.fr, ligia@cbpf.br} )\\
\emph{$^{\mathrm{a}}$
Universit\'e Paris Sud, ISMO, UMR 8214 du CNRS, B\^at. 351, F-91405, Orsay, France} \\
\emph{   $^{\mathrm{b}}$ Centro Brasileiro de Pesquisas Fisicas   } \\
\emph{   $^{\mathrm{c}}$ Instituto Nacional de Ci\^encia e Tecnologia - Sistemas Complexos}\\
\emph{  Rua Xavier Sigaud 150, 22290-180 - Rio de Janeiro, RJ, Brazil  } \\
\emph{  $^{\mathrm{d}}$ Laboratoire APC,
Univ Paris  Diderot, Sorbonne Paris Cit\'e,  }  
\emph{75205 Paris, France   } 
}

\maketitle
{\abstract{ 
In a recent article a generalization of the binomial distribution associated with a sequence of positive numbers was examined. The analysis of the nonnegativeness of the formal expressions was a key-point to allow to give them   a statistical interpretation in terms of probabilities.  In this article we present an approach based on generating functions that solves the previous difficulties: the constraints of nonnegativeness are automatically fulfilled, a complete characterization in terms of generating functions is given and a large number of analytical examples becomes available.
}}

\vspace{0.5cm}
{\bf PACS:}

\tableofcontents

\section{Introduction}
\label{intro}

In most of the realistic models
 of  Physics one must take correlations into account  and events which are usually presented as independent, like in a binomial Bernoulli process, are actually submitted to correlative perturbations. These perturbations lead to deformations of the mathematical independent laws. 
Indeed, 
the deformation of the Poisson distribution upon which is based the construction of Glauber coherent states in quantum optics leads to the so-called nonlinear coherent states 
\cite{dodonov2002,gazeaubook}. 
The realization of a special class of these states, adapted to this deformation, has been proposed in the quantized motion of a trapped atom in a
Paul trap \cite{matosetal1996,kisetal2001}. 
In a recent work \cite{curadoetal2010},  
some of the present authors 
have examined the possibilities of using a certain class of such nonlinear or non-poissonian coherent states in quantum measurements. In the case of imperfect detection,  
the underlying binomial 
distribution was subsequently deformed, raising interesting and non-trivial questions on the statistical content of such sequences.   These questions were examined in a previous note \cite{curadoetal2011}, 
where a binomial-like 
distribution associated to an arbitrary strictly increasing sequence of positive numbers was defined. This distribution was ``formal'' because it was not always positive
 but it was proved that for a certain class of such sequences the proposed 
 binomial-like 
 distribution was nonnegative and possesses a Poisson-like limit law. 
 The question of nonnegativeness in the general case was also discussed and the conditions under which a
probabilistic interpretation can be given to the proposed 
binomial-like 
distribution was examined. 
Other generalizations of the binomial distribution, using different approaches, exist in the literature, see for example \cite{thurner2009,tsallis2012} and references therein. 

In the present note we give a different approach based on generating functions that allows to avoid the previous difficulties: the constraints of nonnegativeness are automatically fulfilled and a great number of analytical examples become available.

In other words we are looking for sequences $x_n$ of positive numbers (except $x_0=0$) such that the associated 
binomial 
deformations remain probability distributions.  
The question is solved in theorem 2.4 where the generating functions associated with these sequences are characterized. 
	
This note is organized as follows. In section \ref{genbernoulli} we present the generalized binomial 
distributions and its generating functions, 
as well as the generalized Poisson distribution; 
 we also develop a comprehensive formalism 
 about some sets of 
sequences that allow a complete probabilistic interpretation of the polynomials from which the distributions are built.  
In section \ref{examples} we present some analytical and numerical examples that illustrate the theoretical framework developed in section \ref{genbernoulli}. 
In section \ref{applications} we discuss possible physical applications of these deformed distributions. In section \ref{conclusions} we present our final comments.  

\section{The generalized binomial distribution and its generating functions}
\label{genbernoulli}
\subsection{The 
binomial-like distribution}
In a process of a sequence of n independent trials with two possibles outcomes, “win” and “loss”, the probability of obtaining $k$ wins is given by the binomial distribution $p_k^{(n)}(\eta)$
\begin{equation}
p_k^{(n)}(\eta) =\left(\begin{array}{c}n \\k \end{array}\right) \eta^k \,(1-\eta)^{n-k}=\dfrac{n!}{(n-k)!\, k!} \,\eta^k \,(1-\eta)^{n-k}  ,
\end{equation}
where the parameter $0 \le \eta \le 1$ is the probability of having the outcome ``win'' and $1-\eta$, the outcome ``loss''.
Therefore, we can assert that the 
binomial distribution above corresponds to the sequence of non-negative integers $n \in \mathbb{N}$.

In a recent article \cite{curadoetal2011} 
it was proposed to generalize the binomial law using an increasing sequence of nonnegative real numbers $\{ x_n \}_{n \in \mathbb{N}}$, where by convention $x_0=0$. The new ``formal probabilities" $\mathfrak{p}_k^{(n)}(\eta)$ are defined as
\begin{equation}
\label{pgotico}
\mathfrak{p}_k^{(n)}(\eta)=\dfrac{x_n!}{x_{n-k}! x_k!} \eta^k p_{n-k}(\eta)  ,
\end{equation}
where the factorials $x_n!$ are given by
\begin{equation}
x_n!=x_1 x_2 \dots x_n, \quad x_0! \eqdef 1,
\end{equation}
and where the polynomials $p_n(\eta)$ are constrainted by the normalization
\begin{equation}
\label{somapgotico}
\forall n \in \mathbb{N}, \quad \sum_{k=0}^n \mathfrak{p}_k^{(n)}(\eta)= 1.
\end{equation}
The polynomials $p_n(\eta)$ can be obtained by recurrence from this relation, the first ones being $p_0(\eta)=1$ and $p_1(\eta)=1-\eta$. The $p_n=\mathfrak{p}_0^{(n)}$ possess a probabilistic interpretation as long as $\mathfrak{p}_k^{(n)}$ are themselves probabilities. This latter condition holds only if the $p_n(\eta)$ are nonnegative. Therefore nonnegativeness of the $p_n$ is a key-point to preserve statistical interpretations. In \cite{curadoetal2011} 
the authors have been able to exhibit one single analytical example ($q$-brackets and their
corresponding $q$-binomial or \textit{Gaussian coefficients}) fulfilling this constraint for all $\eta \in [0,1]$. In the present note we show that an approach based on generating functions allows to completely characterize these cases, giving easily an infinite number of analytical examples. \\
Since the data of a sequence $\{x_n\}_{n \in \mathbb{N}}$ completely specifies the problem in each case, and since we are only interested in the situation where the polynomials $p_n(\eta)$ are always nonnegative for $\eta \in [0,1]$, we introduce the following definition

\begin{defi}
A sequence $\{x_n\}_{n \in \mathbb{N}}$ such that $x_0=0$ and $x_n >0$ (the sequence being increasing or not) is said to be of Complete Statistical Type (CST) if the polynomials $p_n(\eta)$ are 
positive for $\eta \in [0,1[$.
\end{defi}
The purpose of this article is to show the existence of such sequences in 
function of their generating functions as well as to present some non-trivial 
examples.

\subsection{The generating function point of view}
In \cite{curadoetal2011} 
it was proved that the analytical function 
\begin{equation}
\label{eqn:XnGFct}
\mathcal{N}(t)=\sum_{n=0}^\infty \dfrac{t^n}{x_n!}
\end{equation}
gives the generating function of  $p_k(\eta)$ (independently of the 
nonnegativeness 
question) thanks to
\begin{equation}
\label{eqn:PnGFct}
G_{\mathcal{N},\eta}(t) =
\dfrac{\mathcal{N}(t)}{\mathcal{N}(\eta t)}=\sum_{k=0}^\infty \dfrac{t^k}{x_k!} p_k(\eta) .
\end{equation}
Furthermore, introducing the coefficients $I_n$ defined as
\begin{equation}
\label{eqn:InGFct}
\dfrac{1}{\mathcal{N}(t)}=\sum_{n=0}^\infty \dfrac{I_n}{x_n!} (-t)^n,
\end{equation}
the polynomials $p_n$ take the form
\begin{equation}
\label{eqn:PnExpr}
p_n(\eta)=\sum_{k=0}^n \left(\begin{array}{c}x_n \\x_k \end{array}\right) I_k (-\eta)^k,
\end{equation}
where we introduced the notation
\begin{equation}
\left(\begin{array}{c}
x_n \\ 
x_k 
\end{array}\right)=\dfrac{x_n!}{x_k!\,  x_{n-k}!}.
\end{equation}
We deduce from these formula that the knowledge of the analytical functions 
$\mathcal{N}(t)$ of Eq.\eqref{eqn:XnGFct}, $\mathcal{N}(t)^{-1}$ of Eq.\eqref{eqn:InGFct} and $\mathcal{N}(t)/\mathcal{N}(\eta t)$ of Eq.\eqref{eqn:PnGFct} are sufficient to define the coefficients $x_n!$ and the polynomials $p_n$, then the problem can be recast differently only using generating functions. But we first introduce a set $\Sigma_0$ of analytical functions that will be very useful in the sequel.

An interesting remark is that we can also give 
the generating function of the new formal probabilities  $\mathfrak{p}_k^{(n)}$
by the following formula 
\begin{equation}
\dfrac{\mathcal{N}(z \eta t) \mathcal{N}(t)}{\mathcal{N}(\eta t)}=\sum_{n=0}^\infty \dfrac{t^n}{x_n!} \sum_{k=0}^n z^k \mathfrak{p}_k^{(n)}(\eta)\, ,
\end{equation}
where $|t|$ is assumed to be sufficiently small to guarantee the convergence of the series. The proof is straightforward using simple manipulations on the indexes of the sums. 

Let us now introduce a set $\Sigma_0$ of analytical functions that plays a central role in the sequel. 

\begin{defi} In the following $\Sigma_0$ is the set of entire series $f(z)=\sum_{n=0}^\infty a_n z^n$ possessing a non-vanishing radius of convergence and verifying the conditions $a_0=0$, $a_1 >0$ and $\forall n \ge 2, \, a_n \ge 0$.
\end{defi}
The following obvious examples of functions of $\Sigma_0$ are useful in the remainder.

\begin{lem}
\label{lem:SetOm0}
$\forall a,b >0$, $\forall \alpha \in [0,1]$, the functions $F$ defined as $F(t)=e^{a t}-1$, $F(t)=-a \ln(1-b t)$, $F(t)=a \ln \dfrac{1+ \alpha b t}{1-b t}$ belong to $\Sigma_0$.
\end{lem}

\begin{proof}
It is straightforward to show that the series expansions verify the required conditions. 

\end{proof}
\vspace{0.5cm}

\noindent Some simple properties of $\Sigma_0$ are interesting for our purpose:
\begin{prop}
\label{theo:SetOm0}
Properties of $\Sigma_0$
\begin{enumerate}
\item \label{prop1a}  $\forall F,G \in \Sigma_0, F+G \in  \Sigma_0$,

\item \label{prop1b} $\Sigma_0$ is a convex set, 

\item  \label{prop1c} $\forall F \in \Sigma_0, \, \forall \eta \in [-1,1[, \,  t \mapsto F(t)-F(\eta t) \in \Sigma_0 \, {\rm and} \, t \mapsto F(t)+F(-\eta t) \in \Sigma_0$, 

\item \label{pro1d} $\forall F,G \in \Sigma_0, \, \forall a >0, \, (a+F)G \in \Sigma_0$,

\item \label{prop1e} $\forall F,G \in \Sigma_0, \, F \circ G \in \Sigma_0$.

\end{enumerate}
\end{prop}
 
\begin{proof}
The  points (1) and (2) are obvious. The  point (\ref{prop1c}) is due to the identities 
\begin{equation}
F(t)-F(\eta t)= \sum_{n=1}^\infty (1-\eta^n)a_n t^n \quad {\rm  if } \quad F(t)=\sum_{n=0}^\infty a_n t^n
\end{equation}
and
\begin{equation}
F(t)+F(-\eta t)= \sum_{n=1}^\infty (1+(-\eta)^n)a_n t^n \quad {\rm  if } \quad F(t)=\sum_{n=0}^\infty a_n t^n.
\end{equation}
For the  point (4), let us assume $F(t)=\sum_{n=0}^\infty c_n t^n$ and $G(t)=\sum_{n=0}^\infty d_n t^n$, then 
\begin{equation}
H(t)=(a+F(t))G(t)=\sum_{n=0}^\infty u_n t^n \quad {\rm with } \quad u_n=a d_n + \sum_{k=0}^n c_k d_{n-k} .
\end{equation}
Since $a>0 $ and 
$c_k, d_k \ge 0$ for $k \geq 2$ 
then $u_n \ge 0$ for $n \geq 2$. Moreover $u_0=(a+c_0)d_0=0$ because $c_0=d_0=0$, and $u_1=c_1 d_0+(a+c_0)d_1=a d_1 > 0$ because $a,d_1 >0$. Then $H \in \Sigma_0$.\\
For the point (5), assuming the same expansions of $F(t)$ and $G(t)$ as above we have
\begin{equation}
H(t)=(F \circ G)(t)=\sum_{n=0}^\infty c_n G(t)^n
\end{equation}
Since $G(t)=\sum_{n=1}^\infty d_n t^n$ with $d_1 > 0$ and 
$d_n \ge 0$ for all $n \geq 2$, 
the coefficients of $t^k$ in the expansion of $G(t)^n$ are all nonnegative, and then $H(t)=\sum_{n=0}^\infty u_n t^n$ with $u_n \ge 0$. Moreover $H(0)=F(G(0))=F(0)=0$ and $H'(0)=G'(0)F'(G(0))=G'(0)F'(0) >0$. Then we conclude that $H \in \Sigma_0$.
\vspace{0.5cm}
\end{proof}

\begin{coro} 
\label{coro:SetOm0}
From the lemma \eqref{lem:SetOm0} and the  point (\ref{prop1e}) of the proposition  
\ref{theo:SetOm0}
we deduce\\
$\forall F \in \Sigma_0, \, \forall a,b>0, \, \forall \alpha \in [0,1], \, H=e^{a F}-1, -a \ln (1-b F), a \ln \dfrac{1+ \alpha b F}{1-b F}$
belong to $\Sigma_0$.
\end{coro}

\subsubsection{The generating functions set $\Sigma$ for the $x_n!$}
We now define the set $\Sigma$ as

\begin{defi}
$\Sigma$ is the set of entire series $\mathcal{N}(t)=\sum_{n=0}^\infty a_n t^n$ possessing a non-vanishing radius of convergence and verifying $a_0=1$ and $\forall n \ge1, \, a_n>0$. We associate to each of these function $\mathcal{N}$ the nonnegative sequence $\{ x_n \}_{n \in \mathbb{N}}$ defined as $x_0=0$ and $\forall n \ge 1$, 
\begin{equation}
\label{nnseq1}
x_n=\frac{a_{n-1}}{a_n}=n \, \frac{\mathcal{N}^{(n-1)}(0)}{\mathcal{N}^{(n)}(0)}>0 ,
\end{equation}
where $\mathcal{N}^{(n)}(0) = d^n / dt^n \mathcal{N}(t)\vert_{t=0} $.
\end{defi}
With this definition of the coefficients $x_n$ (the sequence $\{x_n\}$ being not necessarily increasing), we obtain $a_n=1/(x_n!)$ and then the definition of Eq.\eqref{eqn:XnGFct} is recovered.

\begin{prop}
\label{theo:SetOm}
 The set $\Sigma$ verifies the following properties:

\begin{enumerate}

 \item \label{prop2a}  $ \forall \mathcal{N}_1,  \mathcal{N}_2 \in \Sigma$, 
 $\mathcal{N}_1 + \mathcal{N}_2 -1  \in \Sigma$, 

\item  \label{prop2b} $\Sigma$ is a convex set,

\item  \label{prop2c} $\forall \mathcal{N}_1, \mathcal{N}_2 \in \Sigma, \, \mathcal{N}_1\mathcal{N}_2 \in \Sigma$,
 
\item  \label{prop2d} $\forall F \in \Sigma_0, \quad t \mapsto e^{F(t)} \in \Sigma$ .
\end{enumerate}
\end{prop}

\begin{proof}
The properties (1), (2) and (3) are obvious from the definition of $\Sigma$. \\
To prove the part (4) of the theorem, let us define $\mathcal{N}(t)=e^{F(t)}$ with $F \in \Sigma_0$. Since $F(0)=0$, we have $\mathcal{N}(0)=1$, then it remains to prove that $\mathcal{N}^{(n)}(0) >0$ for all $n \ge 1$. \\
First $\mathcal{N}'(t)=F'(t) \mathcal{N}(t)$, then $\mathcal{N}'(0)=F'(0) >0$. Moreover for all $k \ge 0$, $F^{(k+1)}(0) \ge 0$. So let us prove by recurrence that for all $n \ge 1$, $\mathcal{N}^{(n)}(t)=G_n(t) \mathcal{N}(t)$ where $G_n$ is an entire series such that $G_n(0) >0$ and $\forall k \ge 1, \, G_n^{(k)}(0) \ge 0$. The property holds true for $n=1$ with $G_1=F'$. Let us assume that it holds true for $n$.
We have $\mathcal{N}^{(n+1)}(t)=(G'_n(t)+G_1(t) G_n(t)) \mathcal{N}(t)$ and so $G_{n+1}=G'_n+G_1G_n$. Now $G'_n(0)+G_1(0) G_n(0)=G_{n+1}(0) >0$ since $G'_n(0) \ge 0$ and $G_1(0), G_n(0) >0$. Finally from $G_{n+1}=G'_n+G_1G_n$ we have obviously $G^{(k)}_{n+1}(0) \ge 0$ since this property holds true for $G_1$ and $G_n$. We conclude that for all $n \ge 1$, $\mathcal{N}^{(n)}(0)=G_n(0) >0$ and then $\mathcal{N} \in \Sigma$.
\end{proof}

Note that the reciprocal of point (4) is not true.
\subsubsection{The generating functions set $\Sigma_+$ for sequences of Complete Statistical Type}
Assuming $\mathcal{N} \in \Sigma$, let us now define the polynomial $p_n(\eta)$ in terms of the generating function $\mathcal{N}(t)/\mathcal{N}(\eta t)$ from Eq.\eqref{eqn:PnGFct} (or using the residue theorem). 

Since each $\mathcal{N} \in \Sigma$ is an entire series verifying $\mathcal{N}(0)=1$ (with a non-vanishing radius of convergence), the function $t \mapsto \mathcal{N}(\eta t)^{-1}$ is also analytical around $t=0$ (for any $\eta$). As a consequence, the function $G_{\mathcal{N},\eta}$ defined as $G_{\mathcal{N},\eta}(t)=\mathcal{N}(t)/\mathcal{N}(\eta t)$ is itself analytical around $t=0$ (for any $\eta$) with $G_{\mathcal{N},\eta}(0)=1$. Therefore the (strict) positiveness of the polynomials $p_n(\eta)$ is in fact equivalent to the assumption $G_{\mathcal{N},\eta} \in \Sigma$. 

As we already know that in fact $p_n(1)=0$ for $n \ne 0$ and $p_0(1)=1$ (using for example $G_{\mathcal{N},\eta}$ for $\eta=1$), we can define the subset $\Sigma_+$ of $\Sigma$ associated with generating functions of CST sequences as

\begin{defi} Using the notation $ G_{\mathcal{N},\eta}(t)=\mathcal{N}(t)/\mathcal{N}(\eta t)$, $\Sigma_+$ is defined as
\begin{equation}
\label{eqn:DefSigma}
\Sigma_+=\left\{ \mathcal{N} \in \Sigma \, | \,  \forall \eta \in [0,1[, \, G_{\mathcal{N},\eta} \in \Sigma \, \right\}
\end{equation}
\end{defi}
\noindent This subset $\Sigma_+$ contains 
deformed binomial distributions with a consistent and complete probabilistic interpretation. As a matter of fact, we have the following properties:

\begin{prop}
\label{theo:SetOmP}
The set $\Sigma_+$ defined in Eq.\eqref{eqn:DefSigma} satisfies the following properties: 
\begin{enumerate}
\item $\forall \mathcal{N}_1, \mathcal{N}_2 \in \Sigma_+, \, \mathcal{N}_1\mathcal{N}_2 \in \Sigma_+$ ,
\item $\forall F \in \Sigma_0, \quad t \mapsto e^{F(t)} \in \Sigma_+$ .
\end{enumerate}
\end{prop}
\begin{proof}
 To prove point (1), 
from the proposition \ref{theo:SetOm} we have 
that $\mathcal{N}_1$ and $\mathcal{N}_2 \in \Sigma$ implies that $\mathcal{N}_1\mathcal{N}_2 \in \Sigma$. Furthermore $G_{\mathcal{N}_1\mathcal{N}_2, \eta}=G_{\mathcal{N}_1, \eta} G_{\mathcal{N}_2, \eta}$, then if $\mathcal{N}_1$ and $\mathcal{N}_2 \in \Sigma_+$, by definition $G_{\mathcal{N}_1, \eta}$ and $G_{\mathcal{N}_2, \eta} \in \Sigma$. Then, from the proposition  \ref{theo:SetOm}, we deduce $G_{\mathcal{N}_1\mathcal{N}_2, \eta} \in \Sigma$.\\
The part (2) is obtained from the proposition \eqref{theo:SetOm0} (point (\ref{prop1c})) and \eqref{theo:SetOm} (point (4)). First if $F \in \Sigma_0$, then $\mathcal{N}=e^F \in \Sigma$ (proposition \eqref{theo:SetOm}). Furthermore $G_{\mathcal{N}, \eta}(t)=e^{F(t)-F(\eta t)}$. For $\eta \in [0,1[$, we know from proposition \eqref{theo:SetOm0} that $t \mapsto F(t)-F(\eta t) \in \Sigma_0$, then we deduce 
from proposition \eqref{theo:SetOm}, $G_{\mathcal{N}, \eta} \in \Sigma$.
\end{proof}

\subsection{The characterization of $\Sigma_+$}
In fact, point (2) in Proposition \ref{theo:SetOmP}  is not only a sufficient but also a necessary condition to obtain functions of $\Sigma_+$. In order to prove it, we start from the following lemma: 

\begin{lem}
\label{lem:NpN}
For any $\mathcal{N} \in \Sigma_+$, the polynomials $p_n$ verify $p'_n(1) \le 0$ and furthermore
\begin{equation}
\label{eqn:NpN}
\dfrac{\mathcal{N}'(t)}{\mathcal{N}(t)}=- \sum_{n=0}^\infty \dfrac{p'_{n+1}(1)}{x_{n+1}!} t^n .
\end{equation}
\end{lem}
\begin{proof}
 To begin with, 
$p_0(\eta)=1$, so $p'_0(\eta)=0$. For $n \ge 1$, the polynomials $p_n(\eta)$ are nonnegative on the interval $\eta \in [0,1[$ and $p_n(1)=0$. Then $p'_n(1)$ cannot be strictly positive, otherwise $p_n(\eta)$ would be negative on some interval $]1-\epsilon,1[$. We conclude that $\forall n \in \mathbb{N}, \, p'_n(1) \le 0$.\\
Now using the function $G_{\mathcal{N},\eta}(t)=\mathcal{N}(t)/\mathcal{N}(\eta t)$ and by a differentiation with respect to $\eta$, we obtain first
\begin{equation}
t \, \dfrac{\mathcal{N}(t) \mathcal{N}'(\eta t)}{\mathcal{N}(\eta t)^2}= - \sum_{n=0}^\infty \dfrac{p'_n(\eta)}{x_n!} t^n .
\end{equation}
Taking into account the special case $p'_0(\eta)=0$ and choosing $\eta=1$ we obtain the equation \eqref{eqn:NpN}.
\end{proof}

\vspace{0.5cm}

This lemma leads to the theorem

\begin{theo}
\label{theo:LnN}
For all $\mathcal{N} \in \Sigma_+$, $\ln \mathcal{N} \in \Sigma_0$ and
\begin{equation}
\label{eqn:LnN}
\ln \mathcal{N}(t)= - \sum_{n=1}^\infty \dfrac{p'_n(1)}{x_n!} \dfrac{t^n}{n}.
\end{equation}
Furthermore we deduce the important characterization
\begin{equation}
\label{characterization1}
\Sigma_+=\{ e^F \, | \, F \in \Sigma_0\}  .
\end{equation}
\end{theo}

\begin{proof}
Eq.\eqref{eqn:LnN} is immediately obtained by integrating Eq. \eqref{eqn:NpN} term by term and  taking into account the value $\ln \mathcal{N}(0)=0$. 
Using the property $p'_n(1) \le 0$, which was shown in lemma \ref{lem:NpN}, and the fact that $p'_1(1)= -1$, as $p_1(\eta)=1-\eta$, we deduce that $\ln \mathcal{N} \in \Sigma_0$.\\ 
The last part of the theorem results from the proposition \eqref{theo:SetOmP}  (point (2)) and the previous comment.
\end{proof}

\vspace{0.5cm}

Theorem \ref{theo:LnN} provides a fair estimate on the manner in which a CST sequence behaves in comparison with natural integers. More precisely, we have: 

\begin{coro}
Let us pick a  $\mathcal{N} \in \Sigma_+$; its associated CST sequence $x_n = n \, \mathcal{N}^{(n-1)}(0)/\mathcal{N}^{(n)}(0)$ then verifies 
\begin{equation}
\label{xnmax}
\forall n \in \mathbb{N}, \, 0 \le x_n \le n x_1 .
\end{equation}
\end{coro}

\begin{proof}
If $\mathcal{N} \in \Sigma_+$, we can find some $F \in \Sigma_0$ such that $\mathcal{N}=e^F$ and we infer from the proof of Proposition \ref{theo:SetOm}  that for all $n \ge 1$, $\mathcal{N}^{(n)}(t)=G_n(t) \mathcal{N}(t)$, where $G_n$ is an entire series such that $G_n(0) >0$ and for $\forall k \ge 1, \, G_n^{(k)}(0) \ge 0$ and $G_{n+1}=G'_n+G_1G_n$. Using Eq. (\ref{nnseq1} this leads to 
\begin{equation}
\forall n \ge 1, \, x_n=n \dfrac{G_{n-1}(0)}{G_{n}(0)}=n \dfrac{G_{n-1}(0)}{G'_{n-1}(0)+G_1(0)G_{n-1}(0)} .
\end{equation}
As $G'_{n-1}(0) \ge 0$, $G_1(0)=1/x_1$, $G_{n-1}(0) >0$ and $x_0=0$, we get Eq.(\ref{xnmax}).
\end{proof}

In the case of odd functions $F\in \Sigma_0$ we obtain more precise results. 

\begin{coro}
Let us assume $\Sigma_+ \ni \mathcal{N}=e^F$ associated to the CST sequence $x_n=n \, \mathcal{N}^{(n-1)}(0)/\mathcal{N}^{(n)}(0)$, with $F \in \Sigma_0$, $F$ being an odd function. Then the (nonnegative) polynomials $p_n$ verify
\begin{equation}
p_n(\eta)=\sum_{k=0}^n \left(\begin{array}{c}x_n \\x_k \end{array}\right) (-\eta)^k
\end{equation}
\end{coro}
\begin{proof}
As $F$ is odd, $\mathcal{N}(t)^{-1}=e^{-F(t)}=e^{F(-t)}$, and so the coefficients $I_k$ previously defined in Eq.\eqref{eqn:InGFct} reduce to $I_k=1$. In that case, from Eq.\eqref{eqn:PnExpr} we immediately obtain the expression above for $p_n$. 
\end{proof}

\subsection{Internal deformations acting on $\Sigma_+$}
In this section we show how natural deformations can be defined using the tools previously presented in subsections 2.2 and 2.3.\\ 
We start with the definition of what we name ``normalized generating functions'' or ``normalized CST sequences''.

\begin{defi}
Let us pick a $\mathcal{N} \in \Sigma_+$. We say that $\mathcal{N}$ is normalized if and only if $\mathcal{N}'(0)=1$; equivalently the associated CST sequence $x_n=n\, \mathcal{N}^{(n-1)}(0)/\mathcal{N}^{(n)}(0)$ is normalized if and only if $x_1=1$. \\
Therefore, to any $\mathcal{N} \in \Sigma_+$ we can associate a normalized function $\mathcal{N}_{\rm norm} (t)=\mathcal{N}(t/\mathcal{N}'(0))$, $\mathcal{N}_{\rm norm} \in \Sigma_+$.\ 

We define the subset $\Sigma_+^{({\rm norm})}$ of $\Sigma_+$ as the set of normalized generating functions.
\end{defi}

\subsubsection{The  $\mathfrak{D}_\alpha$ deformation}
\noindent 
We define the deformation operator $\mathfrak{D}_\alpha$ for $\alpha \in [-1,1[$ acting on $\Sigma_+^{({\rm norm})}$ as

\begin{defi} 
$\forall \alpha \in [-1,1), \, \forall \mathcal{N} \in \Sigma_+^{({\rm norm})}$, 
\begin{equation*}
\mathfrak{D}_\alpha(\mathcal{N})(t)=
\frac{\mathcal{N}(t/(1-\alpha))}{\mathcal{N}(\alpha t/(1-\alpha))}\, .
\end{equation*}
Note that  $\mathfrak{D}_0$ is  the identity .
\end{defi}

\begin{prop} 
For all $\alpha \in [-1,1[$, $\mathfrak{D}_\alpha$ maps $\Sigma_+^{({\rm norm})}$ into $\Sigma_+^{({\rm norm})}$ and all operators $\mathfrak{D}_\alpha$ possess a common fixed point which is the generating function associated with the usual 
Bernoulli trial,  
namely $\mathcal{N}(t)=e^t$.
\end{prop}

\begin{proof}
From the theorem \eqref{theo:LnN}, for any $\mathcal{N} \in \Sigma_+^{({\rm norm})}$, we can find some $F \in \Sigma_0$ such that $\mathcal{N}=e^F$ with $F'(0)=1$ since $\mathcal{N}$ is normalized. Then 
\begin{equation}
\label{dalpha1}
\mathfrak{D}_\alpha (\mathcal{N})(t)=\exp \left( F  \left( \frac{t}{1-\alpha} \right) - F  \left( \frac{\alpha t}{1-\alpha} \right) \right) ,
\end{equation}
From the proposition \eqref{theo:SetOm0} (part (\ref{prop1c})), the function $t \mapsto F  \left( \frac{t}{1-\alpha} \right) - F  \left( \frac{\alpha t}{1-\alpha} \right)$ belongs to $\Sigma_0$ for $\alpha \in [-1,1[$, then from the theorem \eqref{theo:LnN}, $\mathfrak{D}_\alpha (\mathcal{N}) \in \Sigma_+$. Furthermore
\begin{equation}
\mathfrak{D}_\alpha (\mathcal{N})'(0)=\dfrac{1}{1-\alpha} \dfrac{\mathcal{N}'(0)\mathcal{N}(0)-\alpha \mathcal{N}'(0)\mathcal{N}(0)}{\mathcal{N}(0)^2}=1.
\end{equation}
Then $\mathfrak{D}_\alpha (\mathcal{N})$ is normalized and $\mathfrak{D}_\alpha (\mathcal{N}) \in \Sigma_+^{({\rm norm})}$.\\
Verifying that $\mathfrak{D}_\alpha (\mathcal{N})=\mathcal{N}$ for $\mathcal{N}(t)=e^t$ is straightforward.
\end{proof}
\vspace{0.5cm}

The deformation $\mathfrak{D}_\alpha$ allows us to build new explicit CST sequences according to the following proposition:

\begin{prop}
Let us pick some $\mathcal{N} \in  \Sigma_+^{({\rm norm})}$ associated with the normalized CST sequence $x_n=n\, \mathcal{N}^{(n-1)}(0)/\mathcal{N}^{(n)}(0)$ and let $p_n$ be 
the associated polynomials defined by the Eq.\eqref{eqn:PnGFct}. To the new function 
$\mathcal{N}_\alpha=\mathfrak{D}_{\alpha} (\mathcal{N})$ (for $\alpha \in [-1,1[\,$) it corresponds the CST sequence $x_n^{(\alpha)}$ defined as
\begin{equation}
\forall n \ge 1, \, x_n^{(\alpha)} = x_n \dfrac{(1-\alpha)p_{n-1}(\alpha)}{ p_n(\alpha)},
\end{equation}
and the corresponding coefficients $I_n^{(\alpha)}$ of Eq.\eqref{eqn:InGFct} are given by
\begin{equation}
I_n^{(\alpha)}=(-\alpha)^n \dfrac{p_n(1/\alpha)}{p_n(\alpha)}=\dfrac{\sum_{k=0}^n \left(\begin{array}{c}x_n \\x_k \end{array}\right) I_k (-\alpha)^{n-k}}{\sum_{k=0}^n \left(\begin{array}{c}x_n \\x_k \end{array}\right) I_k (-\alpha)^k}.
\end{equation}
The polynomials $p_n^{(\alpha)}(\eta)$ of Eq.\eqref{eqn:PnExpr} become
\begin{equation}
p_n^{(\alpha)}(\eta)=\dfrac{1}{p_n(\alpha)} \sum_{k=0}^n \left(\begin{array}{c}x_n \\x_k \end{array}\right) p_k(1/\alpha) p_{n-k}(\alpha) (\alpha \eta)^k .
\end{equation}
\end{prop}

\begin{proof}
First from Eq.\eqref{eqn:PnGFct} we get  
\begin{equation}
\mathcal{N}_\alpha(t)=\sum_{n=0}^\infty \dfrac{t^n}{x_n!} \dfrac{p_n(\alpha)}{(1-\alpha)^n} = \sum_{n=0}^\infty \dfrac{t^n}{x_n^{(\alpha)}!} ;
\end{equation}
then
\begin{equation}
x_n^{(\alpha)}=x_n \dfrac{(1-\alpha)p_{n-1}(\alpha)}{p_n(\alpha)}.
\end{equation}
Second we have
\begin{equation}
\dfrac{1}{\mathcal{N}_\alpha(t)}=\dfrac{\mathcal{N}(u)}{\mathcal{N}(u/\alpha)} \quad {\rm with } \quad u=\dfrac{\alpha t}{1-\alpha};
\end{equation}
then
\begin{equation}
\dfrac{1}{\mathcal{N}_\alpha(t)}=\sum_{n=0}^\infty \dfrac{u^n}{x_n!} p_n(1/\alpha)=\sum_{n=0}^\infty \dfrac{(-1)^n}{x_n^{(\alpha)}!} I_n^{(\alpha)} t^n.
\end{equation}
By identification term by term, we obtain
\begin{equation}
I_n^{(\alpha)}=(-\alpha)^n \dfrac{p_n(1/\alpha)}{p_n(\alpha)}=\dfrac{\sum_{k=0}^n \left(\begin{array}{c}x_n \\x_k \end{array}\right) I_k (-\alpha)^{n-k}}{\sum_{k=0}^n \left(\begin{array}{c}x_n \\x_k \end{array}\right) I_k (-\alpha)^k}.
\end{equation}
In order to obtain the polynomials $p_n^{(\alpha)}$ we use Eq.\eqref{eqn:PnExpr} and the expressions above. 
\end{proof}

\subsubsection{The  $\mathfrak{E}_\alpha$ deformation}
\noindent 

\begin{defi} 
\label{opdef2}
For all $\alpha \in [-1,1)$ and all $\mathcal{N} \in \Sigma_+^{({\rm norm})}$, the deformation operator $\mathfrak{E}_\alpha$ is defined as 
\begin{equation*}
\mathfrak{E}_\alpha(\mathcal{N})(t)=\mathcal{N}(t/(1-\alpha)) \mathcal{N}(-\alpha t/(1-\alpha))\, .
\end{equation*}
Note that $\mathfrak{E}_0$ is  the identity. 
\end{defi}  

\begin{prop} 
For all $\alpha \in [-1,1[$, $\mathfrak{E}_\alpha$ maps $\Sigma_+^{({\rm norm})}$ into $\Sigma_+^{({\rm norm})}$ and all $\mathfrak{E}_\alpha$ possess a common fixed point which is the generating function associated with the usual Bernoulli trial, 
namely $\mathcal{N}(t)=e^t$.
\end{prop}
\begin{proof}
From the theorem \eqref{theo:LnN}, for any $\mathcal{N} \in \Sigma_+^{({\rm norm})}$, we can find some $F \in \Sigma_0$ such that $\mathcal{N}=e^F$ with $F'(0)=1$ since $\mathcal{N}$ is normalized.  
From definition (2.7), $\mathfrak{E}_\alpha (\mathcal{N})(t)$ is given by 
\begin{equation}
\mathfrak{E}_\alpha (\mathcal{N})(t)=\exp \left( F  \left( \frac{t}{1-\alpha} \right) + F  \left( -\frac{\alpha t}{1-\alpha} \right) \right) .
\end{equation}
From the proposition \eqref{theo:SetOm0} (part (\ref{prop1c})), the function $t \mapsto F  \left( \frac{t}{1-\alpha} \right) + F  \left( -\frac{\alpha t}{1-\alpha} \right)$ belongs to $\Sigma_0$ for $\alpha \in [-1,1[$, then from the theorem \eqref{theo:LnN}, $\mathfrak{E}_\alpha (\mathcal{N}) \in \Sigma_+$. Furthermore
\begin{equation}
\mathfrak{E}_\alpha (\mathcal{N})'(0)=\dfrac{1}{1-\alpha} (\mathcal{N}'(0)\mathcal{N}(0)-\alpha \mathcal{N}'(0)\mathcal{N}(0))=1.
\end{equation}
Then $\mathfrak{E}_\alpha (\mathcal{N})$ is normalized and $\mathfrak{E}_\alpha (\mathcal{N}) \in \Sigma_+^{({\rm norm})}$.\\
 Finally, it is straighforward to verify that 
$\mathfrak{E}_\alpha (\mathcal{N})=\mathcal{N}$ for $\mathcal{N}(t)=e^t$.
\end{proof}
\vspace{0.5cm}

As for $\mathfrak{D}_{\alpha}$, the deformation $\mathfrak{E}_\alpha$ allows to build new explicit CST sequences according to the following  proposition.

\begin{prop}
Let us pick some $\mathcal{N} \in  \Sigma_+^{({\rm norm})}$ associated with the normalized CST sequence $x_n=n\, \mathcal{N}^{(n-1)}(0)/\mathcal{N}^{(n)}(0)$ and the polynomials $p_n$ of Eq.\eqref{eqn:PnGFct}. Let us name $q_n$ the new polynomials defined as
\begin{equation}
\mathcal{N}(t) \mathcal{N}(-\eta t)=\sum_{n=0}^\infty \dfrac{t^n}{x_n!} q_n(\eta).
\end{equation}
The $q_n$ have the explicit expression
\begin{equation}
q_n(\eta)=\sum_{k=0}^n \left(\begin{array}{c}x_n \\x_k \end{array}\right) (-\eta)^k.
\end{equation}
 For $\alpha \in [-1,1)$, to the new function   
$\mathcal{N}_\alpha=\mathfrak{E}_{\alpha} (\mathcal{N})$  there corresponds the CST sequences $x_n^{(\alpha)}$ defined as 
\begin{equation}
\forall n \ge 1, \, x_n^{(\alpha)} = x_n \dfrac{(1-\alpha)q_{n-1}(\alpha)}{ q_n(\alpha)},
\end{equation}
and the polynomials $p_n^{(\alpha)}(\eta)$ 
\begin{equation}
\label{eqn:Pol}
p_n^{(\alpha)}(\eta)=\dfrac{1}{q_n(\alpha)} \sum_{k=0}^n \left(\begin{array}{c}x_n \\x_k \end{array}\right) p_k(\eta) p_{n-k}(\eta) ( -\alpha)^k .
\end{equation}
\end{prop}

\begin{proof}
From Eq.\eqref{eqn:PnGFct}, we have
\begin{equation}
\mathcal{N}_\alpha(t)=\sum_{n=0}^\infty \dfrac{t^n}{x_n!} \dfrac{q_n(\alpha)}{(1-\alpha)^n} = \sum_{n=0}^\infty \dfrac{t^n}{x_n^{(\alpha)}!}
\end{equation}
and so 
\begin{equation}
x_n^{(\alpha)}=x_n \dfrac{(1-\alpha)q_{n-1}(\alpha)}{q_n(\alpha)}.
\end{equation}
Then we have
\begin{equation}
\dfrac{\mathcal{N}_\alpha(t)}{\mathcal{N}_\alpha(\eta t)}= G_{\mathcal{N}, \eta} \left( \dfrac{t}{1-\alpha} \right) G_{\mathcal{N}, \eta} \left( -\dfrac{\alpha t}{1-\alpha} \right) .
\end{equation}
Finally, from the series expansion of $G_{\mathcal{N}, \eta}$ we deduce the expression \eqref{eqn:Pol} for polynomials $p_n^{(\alpha)}$. 
\end{proof}

\subsubsection{Deformed-related transformations in $\Sigma_+^{({\rm norm})}$}
We have established that the deformation $\mathfrak{D}_\alpha$ maps a  deformed exponential $\mathcal{N}(t)$ $\in \Sigma_+^{({\rm norm})}$
into the deformed exponential 

\begin{equation}
\label{rg1}
\mathcal{N}_{\alpha}^{(1)}(t) \equiv \frac{\mathcal{N}\left(\frac{t}{1-\alpha}\right)}{\mathcal{N}\left(\frac{\alpha t}{1-\alpha}\right)} \in \Sigma_+^{({\rm norm})} \, .
\end{equation}
The exponential $\bar{\mathcal{N}}(t)=e^t$ is left invariant under such a transformation for any value of $\alpha$. If we apply again the 
transformation 
$\mathfrak{D}_\alpha$ to $\mathcal{N}_\alpha^{(1)}(t)$ 
we get a new 
deformed exponential $\mathcal{N}_{\alpha}^{(2)}(t)$, also belonging to 
$\Sigma_+^{ ( {\rm norm} ) }$. 
Successive application of $\mathfrak{D}_\alpha$ generate a flow in the set of deformed exponentials belonging to $\Sigma_+^{({\rm norm})}$. 
In order to study the stability of the fixed point $\bar{\mathcal{N}(t)} = e^t$ with respect to $\mathfrak{D}_\alpha$, we recall that a deformed exponential belonging to $\Sigma_+^{({\rm norm})}$ can be written as $e^{F(t)}$, where $F(t) = \sum_{k=0}^\infty f_k t^k$, with $f_0 =0$, $f_1 = 1$ and $f_n \geq 0$ for $n \geq 2$. Applying $\mathfrak{D}_\alpha$ to $\mathcal{N}(t) = \exp\left(F(t)\right)$, see Eq.(\ref{dalpha1}), we obtain the expression 
\begin{equation}
\label{rg2}
\mathcal{N}_{\alpha}^{(1)}(t) = \exp\left(\sum_{k=0}^\infty f_k 
\frac{(1-\alpha^k)}{(1-\alpha)^k} t^k \right) \, .
\end{equation}
Considering each $f_k$ ($\geq 0$) as an axis in an infinite (or finite, if $f_k =0$ for $k \geq k_{max}$)
dimensional space where the transformation acts, we can see that the effect of this transformation is to change each axis  by a factor depending on $k$ and $\alpha$, $f_k \rightarrow f_k \frac{(1-\alpha^k)}{(1-\alpha)^k}$. Two cases have to be considered here: 

\begin{itemize} 
\item  
$-1 \leq \alpha < 0$.  The factor $(1-\alpha^k) / (1-\alpha)^k$ is smaller than 1 for $k \geq 2$ and it is equal to one if 
$k=1$ ($f_0=0$). 
Thus, repeated actions of $\mathfrak{D}_\alpha$ bring $\mathcal{N}(t)$ to  $\bar{\mathcal{N}}(t)= e^t$. Thus the latter  is a 
stable fixed point with respect to $\mathfrak{D}_\alpha$, for the considered range of values of $\alpha$. All the eigenvalues of the transformation at  this fixed point are smaller than one for $k \geq 2$; then, all directions $f_k$ are stable, at the exception of one eigenvalue which is equal to 1 and which  corresponds to  the direction of the ``axis'' $f_1$,  (marginal value and highest eigenvalue).
\item 
$0< \alpha < 1$. The factor $(1-\alpha^k) / (1-\alpha)^k$ is greater than 1 for $k \geq 2$ and it is equal to one if $k=1$ ($f_0=0$). The direction $f_1$ has now the lowest eigenvalue. All the others directions $f_k$, for $k \geq 2$, have increasing eigenvalues as $k$ increases, all of them greater than one (unstable directions). Thus, the fixed point $\bar{\mathcal{N}}(t) = e^t$ is unstable in all directions but one ($f_1$), and 
$\mathfrak{D}_\alpha$ sends the 
deformed exponentials away from the fixed point. In fact, the new fixed points are not element of the 
set $\Sigma_+^{({\rm norm})}$, but instead belong to a more general set $\Sigma_G$.

\end{itemize}
\noindent
Regarding the $\mathfrak{E}_\alpha$ deformation, we can perform 
a similar analysis and we get the same quantitative result for the deformation parameter 
$\alpha$. When $-1 \leq \alpha  < 0$ the fixed point $\bar{\mathcal{N}}(t) = e^t$ is stable with respect to $\mathfrak{E}_\alpha$ transformation and when $0 < \alpha < 1$ it is unstable.

\subsubsection{Generating new functions of $\Sigma_+^{({\rm norm})}$ from a known function of $\Sigma_+^{({\rm norm})}$}
It is possible to construct an arbitrary number of functions belonging to 
$\Sigma_+^{({\rm norm})}$ starting from a known function belonging to $\Sigma_+^{({\rm norm})}$ and
using the following proposition: 
\begin{prop}
Consider a function $\mathcal{N}$ belonging to $\Sigma_+^{({\rm norm})}$. Then, by proposition  (\ref{theo:LnN}),  $\mathcal{N}$ can be written as $\mathcal{N}(t) = e^{F(t)}$, where $F(t)$ belongs to $\Sigma_0$ with $F^\prime(0) = 1$. 
The new functions $F_H(t) = F(t) + H(t)$ where 
$H(t) = \sum_{n=0}^\infty h_n t^n$,   
satisfying $H(0)=0$, 
$H^\prime(0)=0$ 
and having $h_n \geq 0$ for any $n \geq 2$, 
also belong to the set $\Sigma_0$. The new deformed 
exponentials 
$\mathcal{N}_H(t) = e^{F(t) + H(t)}$ belong to $\Sigma_+^{({\rm norm})}$.
\end{prop}

\begin{proof}
As $F(t)$ belongs to $\Sigma_0$  
and generates a $\mathcal{N}(t)$ belonging to $\Sigma_+^{({\rm norm})}$, 
 it can be written as the series 
$\sum_{n=0}^\infty a_n t^n$ where $a_0=0$, $a_1=1$ and $a_n \geq 0$ 
for $n \geq 2$. The function $F_H(t)$ can be written as 
$F_H(t) = \sum_{n=0}^\infty (a_n + h_n) t^n$. As $a_0 + h_0 =0 $,  
$a_1 + h_1 = 1$ and $a_n + h_n \geq 0$ for any $n \geq 2$, this 
implies that $F_H(t)$ belongs to $\Sigma_0$. Then, by proposition  
2.3 the new deformed exponential $\mathcal{N}_H(t)$ belongs 
to $\Sigma_+$. As $\mathcal{N}_H^\prime(0) =1$, it belongs to $\Sigma_+^{({\rm norm})}$.
\end{proof}

\subsubsection{The  $\eta$ deformation}
The expressions given by Eqs. (\ref{pgotico}) and (\ref{somapgotico}) can 
also be deformed by considering a monotonically increasing function $g(\eta)$, with $g(0) =0$ and $g(1)=1$, leading to a deformation of the polynomial $p_{n-k}(\eta)$ when $[0,1] \ni  \eta \mapsto g(\eta) \in [0,1]$. Eq.(\ref{pgotico}) can be rewritten as 
\begin{equation}
\label{pgoticodef}
\mathfrak{p}_k^{(n)}(\eta)=\dfrac{x_n!}{x_{n-k}! x_k!}\,  g(\eta)^k \, p_{n-k}(g(\eta)) \, ,
\end{equation}
satisfying 
\begin{equation}
\label{defsoma}
\sum_{k=1}^n \mathfrak{p}_k^{(n)}(\eta) =1\, ,
\end{equation}
and generating new polynomials $\bar{p}_{n-k}(\eta) \equiv p_{n-k}(g(\eta))$ that 
are positive for $\eta \in [0,1[$  if the polynomials $p_{n-k}(\eta)$
are positive for $\eta \in [0,1[$ (i.e., if they are associated with a sequence of   
 $\Sigma_+$).

\section{Examples}
\label{examples}
The corollary \ref{coro:SetOm0} provides different examples of functions of $\Sigma_0$. Combining this material with  Theorem 2.1 
enables us to assert the following:

\begin{theo}
\label{theo:Example1}
To any $F\in \Sigma_0$ there correspond the following functions $H$ in 
$\Sigma_+$
\begin{enumerate}
\item 
\begin{align*}& H=e^F\, ,
  \exp(e^F-1)\, , \\
   &\forall a\, , \, b>0\,, \, H= (1-b F)^{-a}\, ,   \\
  & \forall a\,,\,b>0\, , \, \forall \alpha \in [0,1]\, , \, H= (1+\alpha b F)^{a} (1-b F)^{-a}\,,\\
  &\forall a\, , \, b>0\,, \, H=(1+b-b e^F)^{-a}\, .   \\
\end{align*}
\item For any sequence $\{a_n \}_{n \in \mathbb{N}}$ of non-negative real numbers such that $\sum_{n=0}^\infty a_n < \infty $, and for any  $\alpha \in [0,1]$, the function $\mathcal{N}$ defined as
\begin{equation}
\label{ngeral}
\mathcal{N}(t)=\prod_{k=0}^\infty \dfrac{1+ \alpha a_k t}{1-a_k t}
\end{equation}
belongs to $\Sigma_+$ (the assumption  $\sum_{k=0}^\infty a_k < \infty $ is needed to obtain a convergent product).
\end{enumerate}
\end{theo}

\begin{proof}
 Point (1) follows  from Lemma \ref{lem:SetOm0} and 
point (2) of Proposition \ref{theo:SetOmP}. Therefore, functions 
$t \mapsto (1+ \alpha a_k t)(1- a_k t)^{-1}$ belong to $\Sigma_+$, 
and it follows from 
Proposition \ref{theo:SetOmP} (point (1))  that 
point (2) holds true. 
\end{proof}

\vspace{0.5cm}

\noindent \textit{Remark} The analytical example from $q$-calculus presented in a previous article \cite{curadoetal2011} is a special case of this theorem corresponding to the choice $a_k=(1-q)q^k$ for $0<q<1$ and $\alpha=0$.

\subsection{Example 1}

Let us consider the function $\mathcal{N}(t) \in \Sigma_+$, given by
\begin{equation}
\label{nex1}
\mathcal{N}(t) = (1-a t)^{-n} \,,
\end{equation}
whose radius of convergence is $t<1/a$, for $a>0$. Function (\ref{nex1}) is obtained from Eq.(\ref{ngeral}) by taking $\alpha = 0$, $a_k=a$ for $0 \leq k \leq n-1 $ and $a_k = 0$ for $k > n-1$. 

The explicit expressions for $x_k^{(n)}!$ are directly calculated by power expansion from 
\begin{equation}
\mathcal{N}(t) = (1-at)^{-n} = \sum_{k=0}^\infty \dfrac{t^k}{x_k^{(n)}!} ,
\end{equation}
giving
\begin{equation}
\label{xk1}
x_k^{(n)}! = \frac{k!}{a^k} \frac{(n-1)!}{(n-1+k)!} \,.
\end{equation}
It follows immediately that
\begin{equation}
\label{xk2}
x_k^{(n)} = \frac{k}{a (n-1+k)}\,. 
\end{equation} 

Notice that when $k \rightarrow \infty$ the coefficients $x_k^{(n)}$ approach the finite value $x_\infty^{(n)} = 1/a$, which does not depend on $n$. Using Eq.(\ref{xk2}) we can deduce a recursion relation for $x_k^{(n)}$, that can be written as: 
\begin{equation}
\label{xkrecrel}
x_{k+1}^{(n)} = \frac{a x_k^{(n)} (n-2) +1   }{a (n- a x_k^{(n)})   }\,.
\end{equation}

Using the expressions for $x_k^{(n)} $ and the generating function, Eq.(\ref{eqn:PnGFct}), 
it is possible to find an explicit expression for the probabilities $p_k(\eta)$, 
valid for $0 \leq \eta \leq 1$:
\begin{eqnarray}
\label{pk1}
p_k^{(n)}(\eta)& = &\sum_{j=0}^{min(k,n)} \binom{n}{j} (-1)^j 
 \frac{(n-1+k-j)!}{(n-1+k)!} \frac{k!}{(k-j)!} \, \eta^j\, \\
 & = & _2F_1(-n,-k;1-k-n;\eta) \, .
\end{eqnarray}
The explicit expressions for the first seven polynomials as well as their figures are shown in the Appendix.

\subsection{Example 2}  
Let us consider $\mathcal{N}(t)$ given by Eq.(\ref{ngeral}) with $\alpha=0$: 
\begin{equation}
\label{ex2}
\mathcal{N}(t) =  \prod_{k=0}^{\infty} \frac{1}{   (1-a_k t)} \,.
\end{equation}
We can divide this example into two classes: either an infinite sequence of $a_k$, all of them different from zero, or a finite sequence, where only a finite number of $a_k$ are  different from zero. 

\subsubsection{\it Finite sequence class} 

As the case with only 
$a_1$ different from zero belongs to the example 1, with $n=1$, 
the simplest case here is when only two $a_k$'s are different from zero, with $0 < a_1 < a_2$ and radius of convergence  
$t < 1/a_1$. The deformed exponential $\mathcal{N}(t)$ is written as 
\begin{equation}
\label{ex2b}
\mathcal{N}(t) = \frac{1}{   (1-a_1 t)} \frac{1}{   (1-a_2 t)} \,
\end{equation}
and we have 
\begin{equation}
\label{xkfatex2}
x_k! = \frac{1}{\sum_{i=0}^{k} a_1^i a_2^{k-i}} \, 
\end{equation}
and
\begin{eqnarray}
\label{*}
x_k &=& \frac{\sum_{j=0}^{k-1} a_1^j a_2^{k-1-j}}{\sum_{i=0}^{k} a_1^i a_2^{k-i}}  = 
\frac{1}{a_2} \frac{(a_1/a_2)^{k} -1}{(a_1/a_2)^{k+1} -1}  \\
\label{xkex2}
&=& \frac{1}{a_2} 
\frac{[k+1]_{a_1/a_2}}{[k+2]_{a_1/a_2}} \, ,
\end{eqnarray}
where $[n]_q \equiv (q^{n+1})/ (q-1)$. 

\noindent
From Eq.(\ref{xkex2}) we can see that the limit of $x_k$ when $k \rightarrow \infty$ is 
equal to $1/a_2$ ($a_2 > a_1$). 

$p_k(\eta)$ is a polynomial in  $\eta$ whose maximum degree is two and its
general expression, for $k \geq 2$, is 
\begin{eqnarray}
p_k(\eta) &=& 1 - \eta \left(1+
\frac{\sum_{i=0}^{k-2} a_1^{i+1} a_2^{k-i-1}    }{\sum_{i=0}^{k} a_1^i a_2^{k-i} }
\right) + \eta^2 
\frac{\sum_{i=0}^{k-2} a_1^{i+1} a_2^{k-1-i} }{\sum_{i=0}^{k} a_1^i a_2^{k-i} } \\
 & = & 1- \eta \left(1+ \frac{1}{a_2} \frac{  [k-1]_{a_1/a_2} }{[k+1]_{a_1/a_2}}   \right) +
 \eta^2 \left(   \frac{1}{a_2} \frac{  [k-1]_{a_1/a_2} }{[k+1]_{a_1/a_2}}   \right)
\, .
\end{eqnarray}
The first polynomials as well as their figures are shown in the Appendix. 

\subsubsection{\it Infinite sequence class} 
In this class all $\{a_k\}$ are different from zero.  In order to satisfy the condition $\sum_i a_i < \infty$ (see Theorem \ref{theo:Example1}) the 
decrease of $a_k$ could be of the power type $a_k \propto 1/k^\alpha$, with $\alpha > 1$, or of the  exponential type on $k$. The latter case  is found precisely in  $q$-calculus  with $q<1$.

\subsection{Example 3: nontrivial case with infinite radius of convergence}
The simplest case with infinite convergence radius other than the exponential (trivial) case, $\mathcal{N}(t) = e^t$,  is  
\begin{equation}
\label{gauss}
\mathcal{N}(t) = \exp \left(t + \frac{a}{2} t^2\right)                   
\end{equation}
with $a>0$. The generating function of the polynomials $p_n(\eta)$ can be written as: 
\begin{equation}
\label{t2}
G_{\mathcal{N},\eta}(t) \equiv \frac{\mathcal{N}(t) }{ \mathcal{N}(\eta t)} = 
\exp \left(     1-\eta) t + \frac{a}{2} (1-\eta^2) t^2      \right) \,.
\end{equation}
The coefficients $x_n$ and the polynomials $p_n(\eta)$ can formally be obtained from Eqs. (\ref{gauss}) and (\ref{t2})
by means of the relations 
\begin{eqnarray}
 x_n & = &  
n \,  \mathcal{N}^{(n-1)}(0)/\mathcal{N}^{(n)}(0) \\
 p_n(\eta) & = & \frac{x_n!}{n!} \,  G_{\mathcal{N},\eta}^{(n)}(0) \, ,
 \end{eqnarray}
where $\mathcal{N}^{(n)}(0)$, $G^{(n)}(0)$ mean the $n$-th derivatives of $\mathcal{N}, G$, respectively, with respect to $t$, evaluated at $t=0$. 

\subsubsection{Relation with Hermite polynomials and calculation of $x_n$}
It is in fact possible to obtain analytical expressions for  the coefficients $x_n$ and the polynomials $ p_n(\eta)$ above. In that sense, let us remember the case of the  Hermite polynomials. One of their generating functions  \cite{magnus,gradstein,abramovitz} is:  
\begin{equation}
\label{genhermite1}
 \exp\left(2 x \tau -  \tau^2 \right) = \sum_{n=0}^\infty H_n(x) \frac{\tau^n}{n!} \, .
\end{equation}
Making $\tau = i \sqrt{a/2} \, t$ and $x = -i / \sqrt{2 a}$, it is possible 
to formally write $\mathcal{N}(t)$ given by Eq.(\ref{gauss}) as: 
\begin{equation}
\label{defexp2}
\mathcal{N}(t) = \exp\left(t + \frac{a}{2} t^2 \right) = \sum_{n=0}^\infty 
H_n \left( \frac{-i}{\sqrt{2 a}}  \right) \frac{i^n \left(\frac{a}{2}\right)^{n/2}}{n!} \, t^n \, .
\end{equation}
This implies that the factorial, $x_n!$, can be written, formally, as: 
\begin{equation}
\label{xnfat2} 
x_n! = \left[  \frac{i^n \left(\frac{a}{2}\right)^{n/2}}{n!} H_n \left( 
\frac{-i}{\sqrt{2 a}}
\right) \right]^{-1} \,.
\end{equation}
The Hermite polynomial has a formal expansion given by 
(see \cite{magnus}, p.249 or \cite{abramovitz}, p.772):
\begin{equation}
\label{Hermexp} 
 H_n (x) = n! \sum_{m=0}^{\left[ n/2  \right] } \frac{(-1)^m \,(2 x)^{n-2 m} }
 { m! \, (n-2 m)!  } \, ,
\end{equation}
where $[n/2]$ 
is the greatest integer less than or equal to $n/2$ (floor function).
Substituting Eq.(\ref{Hermexp}) in Eq.(\ref{xnfat2}) and putting $x = -i / \sqrt{2a}$ we have the analytical expression for $x_n!$ in function of $a$,  
\begin{equation}
\label{xnfat3} 
 x_n! = \left[  \sum_{m=0}^{[n/2]} 
 \frac{ \left(\frac{a}{2}\right)^{m} } {m! \, (n-2m)!}  .
\right]^{-1} \,,
\end{equation}
 
The sequence element 
$x_n = x_n! / x_{n-1}!$ is then found to be: 
\begin{equation}
\label{xn2} 
 x_n = \frac{x_n!}{x_{n-1}!} = \frac{
   \sum_{j=0}^{[(n-1)/2]} 
 \frac{ \left( a/2\right)^{j} } {j! \, (n-1-2j)!} }
 { \sum_{m=0}^{[n/2]} 
 \frac{ \left( a/2 \right)^{m} } {m! \, (n-2m)!}} \, .
\end{equation}

\subsubsection{Asymptotic behavior of $x_n$}

In order to get the asymptotic behavior of $x_n$, we note that the known Hermite polynomials recurrence relation \cite{magnus,abramovitz,gradstein}, 
\begin{equation}
\label{hermrecu2}
H_{n+1}(x) = 2 x H_{n}(x) - 2n H_{n-1}(x) \, ,
\end{equation}
when used in Eqs.(\ref{xnfat2}) and (\ref{xn2}), yields a simple recursion relation for the coefficients $x_n$ 
given by
\begin{equation}
\label{xnrecu2}
x_{n+1} = \frac{n+1}{1+a x_n} \, . 
\end{equation}
With the assumption that $x_n$ tends to infinity when $n \rightarrow \infty$, 
we are led 
to $x_n^2 \approx n/a$, from which follows the asymptotic behavior, 
valid for $n$ even or odd, 
\begin{equation}
\label{xnlim} 
x_n \sim \sqrt{\frac{n}{a}} + \cdots \,.
\end{equation}

\noindent
On the other side, the hypothesis that $x_n$ tends to a finite value when $n \rightarrow \infty$ leads to a contradiction. Therefore, in the case of Hermite polynomials the sequence $x_n$ goes to infinity as $\sqrt n$, unlikely the other two examples, $1$ and $2$, shown above, for which $x_n$ tends to a finite value as $n \to \infty$. 

\vspace{0.5cm}

Due to the previous results and to numerical verifications up to large values of $n$ we can propose the following conjectures. 

\noindent
{\bf Weak conjecture:} 

\noindent
{\it Conjecture 1:} In the case of the deformed exponential 
$\mathcal{N}(t) = \exp(t + (a/3) t^3)$ we also have a simple recursion relation for the new $x_n$ coefficients, given by 
\begin{equation}
\label{xnrecu3}
x_{n+1} = \frac{n+1}{1+a x_n x_{n-1}} \, , 
\end{equation}
leading to the asymptotic behavior $x_n \sim (n/a)^{1/3}$. 
A natural conjecture, that has been numerically verified, is that 
a general deformed exponential $\mathcal{N}(t) = \exp(t + (a/m) t^m)$ has coefficients 
satisfying the recursion relation $x_{n+1} = (n+1) / (1+ a x_n x_{n-1}   \cdots 
 x_{n-m+2} )$
and having the asymptotic behavior $x_n \sim (n/a)^{1/m}$ when $n \rightarrow \infty$. 

\vspace{0.5cm}
\noindent
{\bf Strong conjecture:} \\
We have also verified numerically 
a strong version of the previous conjecture that can be formulated as: 

\noindent
{\it Conjecture 2:} The deformed exponential $\mathcal{N}(t) = \exp(\sum_{n=1}^m a_n t^n)$, with $a_1 =1$ and $a_n > 0 $ for all $n \geq 2$,
if written as $\mathcal{N}(t) = \sum_{n=1}^\infty  t^n/x_n! $
 has a simple recursion relation for the new $x_n$ coefficients, given by 
\begin{equation}
\label{xnrecu4}
x_{n+1} = \frac{n+1}{1+ a_2 x_n + a_3 x_n x_{n-1} + a_4 x_n x_{n-1} x_{n-2} + \cdots + a_m x_n x_{n-1}  \cdots x_{n-m+2}
 } \, , 
\end{equation}
leading to the asymptotic behavior $x_n \sim (n/a_m)^{1/m}$ when $n \rightarrow \infty$. 

\subsubsection{Polynomials}

From the generating function of the deformed polynomials $p_n(\eta)$ given by Eq.(\ref{t2}), by means of the substitution
\begin{eqnarray}
\tau  & = & i \sqrt{\frac{a}{2} (1-\eta^2)} \, t  \\
x  & = & -i \, \sqrt{ \frac{ 1-\eta}{2  a (1+\eta)}        } \,,
\end{eqnarray}
we obtain
\begin{equation}
 \exp\left(2 x \tau -  \tau^2 \right) =  G_{\mathcal{N},\eta}(t) \, ,
\end{equation}
and therefore the generating function can be identified with the Hermite polynomial through Eq. (\ref{genhermite1}). 
Eq.(\ref{eqn:PnGFct}) allows us to write for the polynomials $p_n(\eta)$ the expression  
\begin{equation}
\label{polygeral}
p_n(\eta) =  i^n \, \left(  \frac{a}{2} (1-\eta^2)   \right)^{n/2} \,  \frac{x_n!}{n!} \, 
H_n \left(-i \, \sqrt{      \frac{ 1-\eta}{2  a (1+\eta) }     } \,\,  \right) \, , 
\end{equation} 
where $x_n!$ is given by Eq.(\ref{xnfat3}).  
Using Eqs.(\ref{Hermexp}) and (\ref{xnfat3}) we can write a general explicit formula for the deformed polynomials generated by Eq.(62):
\begin{equation}
\label{polygeral2}
p_n(\eta) =\left[ 
\frac{a}{2} (1-\eta^2) 
\right]^{n/2}  \,
\frac{\sum _{j=0}^{ [n/2]} 
\left(
\frac{[ 2 (1-\eta )/(a (1+\eta ) )  ]^{n/2 - j}   }{j! (n-2 j)!}
\right)}
{\sum _{m=0}^{[n/2]} 
\left(\frac{( a/2 )^{m}}{m! (n-2 m)!} 
\right)} \,.
\end{equation}
Therefore, Eqs.(\ref{xnfat3}), (\ref{xn2}) and (\ref{polygeral}) give 
the explicit expressions for $x_n!$, $x_n$ and $p_n(\eta)$. 
The first coefficients $x_n$, the polynomials $p_n(\eta)$, and respective figures are shown in 
the Appendix.

\section{Possible applications}
\label{applications}

The formalism developed in this article 
 can be useful in many areas of physics and in other scientific fields. 
 The usual binomial 
 and Poisson laws are in fact strictly obeyed by systems where correlations are absent.  But in physical processes where the correlations are strong enough, which are processes found everywhere, those laws are not applicable. We claim that our formalism naturally takes these correlations into account. 
Strong correlations also change the average values of 
the probability of having ``win'' or ``loss'', and a discussion about how these averages are modified can be found in \cite{curadoetal2011}. 

As an example, 
our formalism could be applied to quantum optics and atomic physics, where the  Poisson distribution, on which the construction of Glauber coherent states is based, can be 
found experimentally deformed.
Such deformations may be of sub-Poissonian or super-Poissonian type 
\cite{dodonov2002,matosvogel1996,vogeletal2001,loudon}. We know 
that the deformations of the Poissonian law can be associated with the so-called nonlinear coherent states and can be represented by the deformed 
exponentials $\mathcal{N}(t)$, extensively discussed in this work. A special 
type of nonlinear coherent states, corresponding to various deformed exponentials, was actually associated with the trapped-atom motion
\cite{matosvogel1996,vogeletal2001}. 

Another example is found in quantum measurement, where the information carriers can be associated with quantum states and measurements with operators. If the quantum states carrying information are not orthogonal, no measurement can distinguish between overlapping quantum states without some ambiguity, implying that errors are unavoidable. On the other side, 
codifying information by means of nonorthogonal quantum states has some advantages, as 
it is known that a classical information capacity of a noisy channel is actually maximized by a nonorthogonal alphabet. This  
justifies the development of a quantum information  
formalism 
based on nonorthogonal states. It is known that when one tries to distinguish between two nonorthogonal states through some receiver device, there always exists a quantum error probability, whose quantum limit is given by the so-called Helstrom bound, which is the smallest physically allowable error probability, taking into account the overlap between two states. 
In an alphabet consisting only of two words (two coherent states, for example), 
 the expression for the quantum error probability (Helstrom bound) of two overlapping states $ \vert \Psi_0 \rangle$ and $ \vert \Psi_1 \rangle$ 
 is given by $P_H = (1/2) \left(  
1-\sqrt{
1- \vert \langle\Psi_1 \vert \Psi_0 \rangle \vert^2  
 } 
\right)$, 
if the probabilities of the sender to transmit the message associated with $ \vert \Psi_0 \rangle$ and $ \vert \Psi_1 \rangle$  are equal.  
It is also 
known that an imperfect detection, due to a non-ideal photodetector, which is always the case in real measurements, modifies the theoretical probability to detect $n$ photons. Now  
 a non-Poissonian law could be fitted with a deformed exponential $\mathcal{N}(t)  = \sum_n t^n / x_n!$ (associated with a non-Poissonian law $ t^n / (\mathcal{N}(t) x_n!) $) so that we deal with a sequence of nonnegative numbers $(x_n)$  being viewed as a phenomenological spectrum. Formally, to such a phenomenological spectrum are naturally associated nonlinear coherent states given by $ | z \rangle = (1/\mathcal{N}(t)) \sum_n z^n/x_n!  |n\rangle $.  
These phenomenologically-constructed nonlinear coherent states, adjusted with the 
effective numbers $(x_n)$,  
can be set side-by-side with a corresponding deformed exponential, or a non-Poissonian law. The Helstrom bound shown above, associated with the nonlinear coherent state, is modified, and can be read now as, 
$P_H^{({\rm nonlin})} = (1/2) \left(  
1-\sqrt{
1- 1/\mathcal{N}(t)  
 } 
\right)$, 
where $\mathcal{N}(t)$ is the deformed exponential associated with the sequence of data measured (a sub or super-Poissonian distribution for example) and $\eta \in [0,1]$ is the efficiency of the detector. 
It exists, then, the theoretical possibility of lowering the Helstrom bound with respect to the 
value 
obtained using Glauber coherent states, which could be very useful in a  transmission of information. 

Beyond the examples discussed above, one should point out that the formalism developed in this article was not constructed for a specific physical system, but is very general. It could be applied everywhere, 
under the condition that important correlations are present. The examples above are examples from physics, but it is not difficult to find examples in other areas where deviations from 
Bernoulli trial 
and/or Poisson distributions have to be considered.  

\section{Conclusion}
\label{conclusions}

In this article we have developed a 
formalism 
allowing us to construct deformed 
binomial 
and Poisson distributions preserving the probabilistic interpretation intrinsically associated with these distributions. We have shown mathematically, (see theorem 2.1), under which conditions these deformations can be constructed and we have illustrated our results 
with  some key examples. These deformed distributions can be associated with examples coming from quantum optics, atomic physics and quantum information but the formalism enhanced in this article is very general, and could be used in many different domains where deviations of 
binomial 
and/or Poisson laws are observed.   In spite of the fact that we have explored practically all the important aspects of the set $\Sigma_+$, it remains the question of the existence of a larger set, containing this one, and still keeping the probabilistic interpretation, so important to the phenomenological use of the  
formalism. 
Analysis in this direction are underway.

\section*{Acknowledgments}
EMFC acknowledges the 
partial financial supports by CNPq, CAPES and FAPERJ (Brazilian scientific 
agencies).

\newpage 

\section*{Appendix}

\subsection*{Example 1}
The 7 first polynomials if $n=5$ are: 
\begin{eqnarray}
p_1^{(5)}(\eta)& = &1-\eta \\
p_2^{(5)}(\eta) & = & 1-\frac{5}{3} \eta + \frac{2}{3} \eta^2 \\
p_3^{(5)}(\eta) & = & 1-\frac{15}{7} \eta + \frac{10}{7} \eta^2 - \frac{2}{7} \eta^3 \\
p_4^{(5)}(\eta) & = & 1-\frac{5 \eta }{2}+\frac{15 \eta ^2}{7}-\frac{5 \eta ^3}{7}+\frac{\eta ^4}{14} \\
p_5^{(5)}(\eta) & = & 1-\frac{25 \eta }{9}+\frac{25 \eta ^2}{9}-\frac{25 \eta ^3}{21}+\frac{25 \eta ^4}{126}-\frac{\eta ^5}{126} \\
p_6^{(5)}(\eta) & = & 1-3 \eta +\frac{10 \eta ^2}{3}-\frac{5 \eta ^3}{3}+\frac{5 \eta ^4}{14}-\frac{\eta ^5}{42} \\
p_7^{(5)}v & = & 1-\frac{35 \eta }{11}+\frac{42 \eta ^2}{11}-\frac{70 \eta ^3}{33}+\frac{35 \eta ^4}{66}-\frac{\eta ^5}{22} \,.
\end{eqnarray}
These polinomials can be found in figure \ref{fig1}.
Notice that the higher possible power of $\eta$ in $p_k^{(n)}(\eta)$ is $n$ (here $n=5$), which happens in all cases where $k \geq n$. 
 
\begin{figure}
\begin{center}
\includegraphics[width=4in]{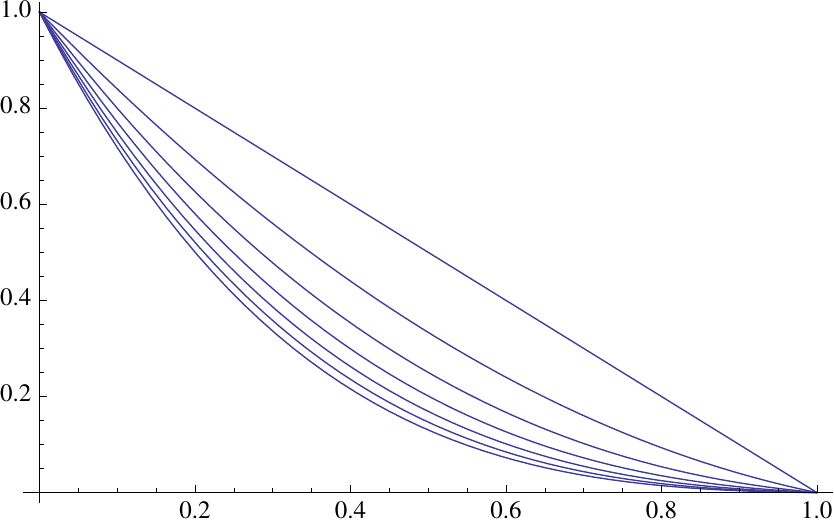}
\caption{Polynomials $p_k^{(n)}(\eta)$ for $n=5$ and $k=1, 2, 3, 4, 5, 6, 7$. The values of $k$ increase from top to bottom.}
\label{fig1}
\end{center}
\end{figure}

\subsection*{Example 2}

\noindent
The first three polynomials are 
\begin{eqnarray}
p_1(\eta)& = &1-\eta \\
p_2(\eta) & = & 1- \eta (1+ b_1) + \eta^2 b_1 \\
p_3(\eta) & = & 1- \eta (1+b_2) + \eta^2 b_2 \, ,
\end{eqnarray}
where 
$$
b_1 \equiv \frac{a_1 a_2}{a_1^2 + a_1 a_2 + a_2^2} 
\, ,
$$
and  
$$
b_2 = \frac{a_1^2 a_2 + a_1 a_2^2}{a_1^3 + a_1^2 a_2 + a_1 a_2^2+a_2^3} \, .
$$

The polynomials for $k=1, 2, 3, 4$ are shown in figures 2, 3 and 4, for the values of 
$(a_1,a_2)=(1/4,1/2), (1/3,4/3)$ and $(5/4,6/4)$ respectively. 
\begin{figure}
\begin{center}
\includegraphics[width=4in]{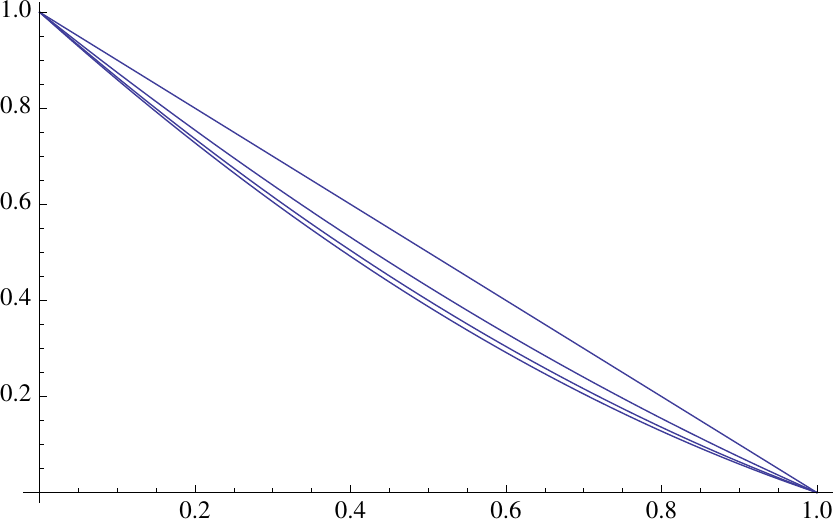}
\caption{Polynomials $p_k(\eta)$ for $k=1, 2, 3, 4$ and 
$(a_1,a_2)=(1/4,1/2)$. The values of $k$ increase from top to bottom.}
\label{fig2}
\end{center}
\end{figure}

\begin{figure}
\begin{center}
\includegraphics[width=4in]{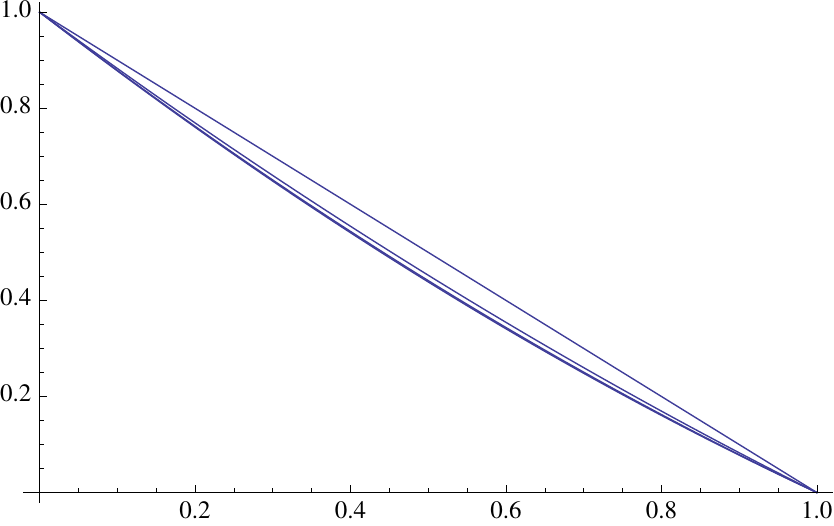}
\caption{Polynomials $p_k(\eta)$ for $k=1, 2, 3, 4$ and $(a_1,a_2)=(1/3,4/3)$. The values of $k$ increase from top to bottom.}
\label{fig3}
\end{center}
\end{figure}

\begin{figure}
\begin{center}
\includegraphics[width=4in]{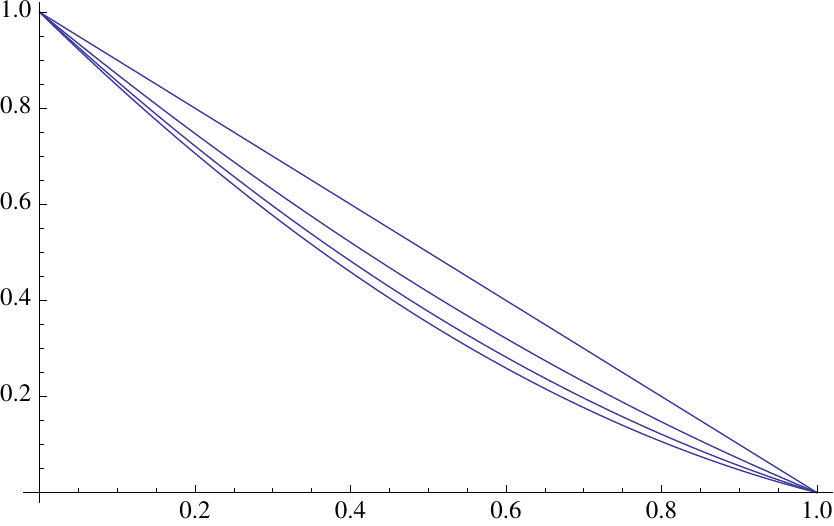}
\caption{Polynomials $p_k(\eta)$ for $k=1, 2, 3, 4$ and $(a_1,a_2)=(5/4,6/4)$. The values of $k$ increase from top to bottom.}
\label{fig4}
\end{center}
\end{figure}

\subsection*{Example 3}

With the formulas given by 
Eqs.(\ref{xnfat3}), (\ref{xn2}) and (\ref{polygeral}) 
 we can calculate the coefficients and polynomials. The
first five coefficients $x_n$ are: 
$x_0=0, x_1=1, x_2 = 2/(1+a), x_3 = 3 a / (1+ 3 a), x_4 = 4(1+3a)/(1+6a + 3 a^2), x_5 = 5(1+6a+3 a^2)/(1+10a+15 a^2)$ and the first polynomials $p_n(\eta )$ for $n = 2, 3, 4$ and $5$ are (as in all the other examples, $p_0(\eta)=1$ and $p_1(\eta) = 1-\eta$): 
\begin{eqnarray}
p_2(\eta) & = & 1-\frac{2  }{1+a} \eta+\frac{(1-a)}{1+a}  \eta^2 \\
p_3(\eta) & = & 1-\frac{3(1+ a) }{1+3 a}  \eta +\frac{3(1- a) }{1+3 a}\eta^2 
- \frac{(1-3 a) }{1+3 a} \eta^3 \\
p_4(\eta) & = &  1-\frac{4(1+3 a)  }{1+6 a+3 a^2}\eta + 
\frac{6 \left(1- a^2\right) }{1+6 a+3 a^2}\eta^2  
- \frac{4(1-3 a) }{1+6 a+3 a^2}\eta^3  \\
  &  &  + \,  \frac{\left(1-6 a+3 a^2\right) }{1+6 a+3 a^2}\eta^4 \\
p_5(\eta) & = & 
1- \frac{5\left(1+6 a+3 a^2\right)  }{1+10 a+15 a^2}\eta + 
\frac{10 \left(1+2 a-3 a^2\right) }{1+10 a+15 a^2} \eta^2  \\
&& - \, \frac{10 \left(1-2 a-3 a^2\right) }{1+10 a+15 a^2} \eta^3 + 
\frac{5 \left(1-6 a+3 a^2\right) }{1+10 a+15 a^2}\eta^4  \\
&& - \, \frac{ \left(1-10 a+15 a^2\right) }{1+10 a+15 a^2}\eta^5   \, .
\end{eqnarray}
Notice that for $a>0$ the coefficients $x_n$ are always smaller than $n$ (since $x_1 = 1$), 
as shown theoretically before in corollary 2.2, Eq. (\ref{xnmax}). The polynomials for $n=1, 2,3, 4, 5$ are plotted in Fig. \ref{figt2} for a=1/2.

\begin{figure}
\begin{center}
\includegraphics[width=8cm]{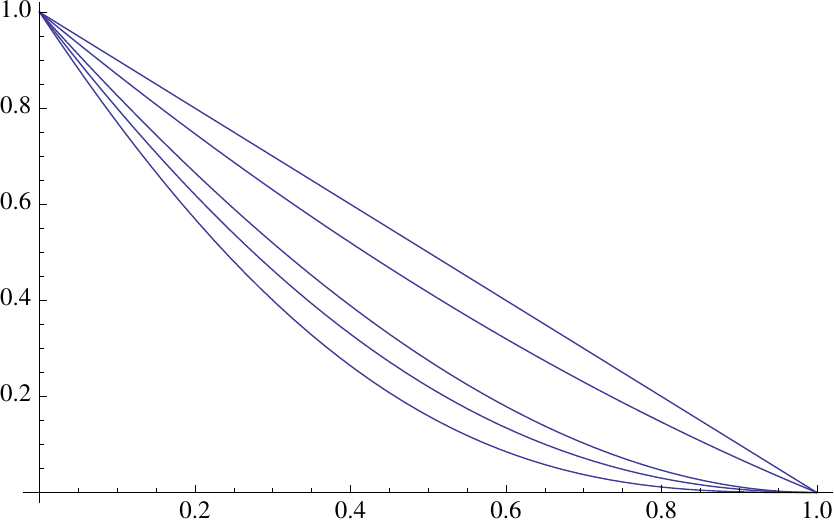}
\caption{Polynomials ($n = 1, 2, 3, 4, 5$) for $\mathcal{N}(t) = \exp(t + a t^2/2)$,  with $a=1/2$. $n$ increases from top to bottom. }
\label{figt2}
\end{center}
\end{figure}

\end{document}